\documentclass[letterpaper, 10pt, conference]{ieeeconf}      

\IEEEoverridecommandlockouts                              
                                                          
\usepackage{amsmath,amssymb,euscript,yfonts,psfrag,latexsym,dsfont,graphicx}
\usepackage{bbm,color,amstext,wasysym,parskip,balance}
\usepackage{subcaption}
\usepackage{balance}
\usepackage{url}
\graphicspath{{./},{./figures/},{./Matlab/}}

\newtheorem{thm}{Theorem}

\newtheorem{prop}{Proposition}

\newtheorem{remark}[thm]{Remark}

\newcommand{\mR}{{\mathbb R}}

\newcommand{\bx}{{\mathbf x}}
\newcommand{\mE}{{\mathbb E}}

\newcommand{\cM}{{\mathcal M}}

\newcommand{\tr}{\operatorname{trace}}

\newcommand{\trace}{\operatorname{tr}}

\newcommand{\jk}[1]{{\color{magenta}{#1}}}

\definecolor{grey}{rgb}{0.6,0.6,0.6}
\definecolor{lightgray}{rgb}{0.97,.99,0.99}

\begin{document}

\title{State tracking of linear ensembles via optimal mass transport}


\author{Yongxin Chen and Johan Karlsson
\thanks{*This work was in part supported by the Swedish Research Council (VR) grant 2014-5870 and the Swedish Foundation of Strategic Research (SSF) grant AM13-0049.}
\thanks{Y.\ Chen is with the Department of Electrical and Computer Engineering, Iowa State University, email:yongchen@iastate.edu. J.\ Karlsson is with Department of mathematics, KTH Royal Institute of Technology. email: johan.karlsson@math.kth.se.
}}

\maketitle

\begin{abstract}

We consider the problems of tracking an ensemble of indistinguishable agents with linear dynamics based only on output measurements. In this setting, the dynamics of the agents can be modeled by distribution flows in the state space and the measurements correspond to distributions in the output space.
In this paper we formulate the corresponding state estimation problem using optimal mass transport theory with prior linear dynamics, and the optimal solution gives an estimate of the state trajectories of the ensemble.
For general distributions of systems this can be formulated as a convex optimization problem which is computationally feasible with  when the number of state dimensions is low.
In the case where the marginal distributions are Gaussian, the problem is reformulated as a semidefinite programming and can be efficiently solved for tracking systems with a large number of states.




\end{abstract}

\section{Introduction}

The optimal mass transport theory provides a geometric framework for mapping a distribution to another one in a way that minimizes the total transport cost \cite{villani2003topics}. 
This has been used in many contexts, traditionally for application in economics and logistics \cite{Gal16}, and more recently in imaging and machine learning \cite{benamou2015iterative, engquist2014application, haker2004optimal,karlsson2016generalized} as well as systems and controls \cite{CheGeoPav14a,CheGeoPav17a, elvander2017interpolation}. 
In case when the transport cost is quadratic, the problem may be formulated as a fluid dynamics problem \cite{benamou2000computational}. 
It can also be viewed as an optimal control problem of the density of the particles that obey the dynamics $\dot x(t)=u(t)$ \cite{chen2016relation}. In the subsequent paper \cite{chen2017optimal} a natural generalization of this problem is introduced where the underlying linear dynamics of the particles become $\dot x(t)=Ax(t)+Bu(t)$.

In this work we consider the extension of this framework to the case where the full state information is not available. 
In particular we consider the tracking problems where we seek to estimate the states of several identical and indistinguishable systems from only their joint outputs. 
This is also known as state estimation of ensembles, see \cite{zeng2016ensemble,zeng2017sampled}. 
One of the main obstacles is that it is not known which output that is generated by a certain subsystem, hence an association problem has to be solved. A brute force approach to this would result in a combinatorial problem. However, by formulating this as an optimal transport problem which is convex and where the number of variables only grow linearly in the number of stages.
Our formulation not only allows for tracking a finite number of particles, but also applies to the more general problems of tracking distributions.
We also consider the case when the underlying distributions are Gaussian. In this case the number of variables can be reduced significantly and the problem can be solved efficiently with large number of state dimensions.

At the very high level, we are developing a framework to smoothly interpolate a sequence of probability densities, in a way akin to smoothly interpolate several points in the Euclidean space. Indeed, the cubic spline interpolation of several points has a variational formulation in the flavor of optimal control \cite{Hol57}. The cubic spline counterpart in the space of distributions has been recently studied in \cite{CheConGeo18,BenGalVia18}. Our work can be viewed as a generalization of these where more general underlying linear dynamics, instead of simple second order integrator, are considered. 

The rest of the paper is structured as follows. Section \ref{sec:background} is a brief introduction to the optimal mass transport theory. In Section \ref{sec:track}, we formulate the tracking problems using optimal mass transport theory. The case where the densities are assumed Gaussian is discussed in Section \ref{sec:gaussian}. Finally, we present several numerical examples in Section \ref{sec:example} to illustrate our framework.

\section{Background on optimal mass transport}\label{sec:background}
Monge's original formulation of optimal mass transport is as follows (see, e.g., \cite{villani2003topics}). Consider two nonnegative distributions, $\rho_0$ and $\rho_1$, of the same mass, defined on a compact set $X\subset \mR^n$. The optimal mass transport problem seeks a transport function $f: X\to X$ 
that minimizes the transportation cost
\[
\int_{X}c(x,f(x))\rho_0(x)dx
\]
over all the mass preserving maps from $\rho_0$ to $\rho_1$, namely,  
\begin{equation}\label{eq:masspreserving}
\int_{x\in \mathcal{U}}\rho_{1}(x)dx=\int_{f(x)\in \mathcal{U}}\rho_{0}(x)dx\quad \mbox{ for all } \mathcal{U}  \subset X,
\end{equation}
which is often denoted $f_\# \rho_0=\rho_1$. 
The function $c(x_0,x_1):X\times X\to \mR_+$ is a cost function that describes the cost for transporting a unit mass from $x_0$ to $x_1$. The Monge's problem is usually difficult to solve due to the nontrivial 
constraint $f_\# \rho_0=\rho_1$ and the minimum may not exist. To overcome these difficulties, Kantorovich proposed a linear programming  relaxation 
	\begin{equation}
		\min_{\pi \in \Pi(\rho_0,\rho_1)} \int_{X\times X} c(x_0,x_1) d\pi(x_0,x_1)
	\end{equation}
where $\Pi(\rho_0,\rho_1)$ denotes the set of all joint distributions between $\rho_0$ and $\rho_1$. In fact, when $\rho_0$ and $\rho_1$ are absolutely continuous, 
these two formulations are equivalent.
	
When the cost function is quadratic, i.e., $c(x_0,x_1)=\|x_0-x_1\|^2_2$, the optimal mass transport problem can be set up as an optimal control problems in fluid dynamics \cite{benamou2000computational}
\begin{align*}
\min_{u,\hat\rho} \quad & \quad\int_{0}^1\int_{x\in X} \|u(t,x)\|^2\hat\rho(t,x)dxdt\\
\mbox{subject to } & \quad \frac{\partial\hat\rho}{\partial t}+\nabla \cdot (u\hat\rho)=0\\
&\quad \rho_k=\hat\rho(k,\cdot), \; \mbox{ for } k=0,1. 
\end{align*}
 This can be interpreted as an optimization problem where the mass distributions represented by infinite-decimal particles, each carrying a cost corresponding to the optimal control problem   
\begin{align*}
	\min_{u} \qquad & \int_{0}^{1} \|u(t)\|^2 dt \\
	\mbox{subject to }\quad  & \dot x(t) = u(t), \\
	 &x(0) = x_0 \mbox{ and } x(1)=x_1
\end{align*}
where $x_0, x_1$ are the initial and final position of the particle, respectively. Hence, choosing the quadratic cost in the optimal mass transport problem can be seen as assuming the underlying dynamic being $\dot x(t) = u(t)$. 

In \cite{chen2017optimal} this was generalized through replacing the cost function that reflects deviation from the trajectory of the underlying system dynamics. 
It is associated with the linear dynamic  
\begin{eqnarray}
	\dot x(t)&=&A x(t) +Bu(t)
\end{eqnarray}
and optimal control problem
\begin{align*}
	\min_{u} \qquad & \int_{0}^{1} \|u(t)\|^2 dt \\
	\mbox{subject to }\quad  & \dot x(t) = Ax(t)+Bu(t), \\
	 &x(0) = x_0 \mbox{ and } x(1)=x_1.
\end{align*}
The cost is then given by 
\begin{equation}\label{eq:cost_lineardyn}
	c(x_0,x_1) =(x_1-\Phi  x_0)^T Q(x_1-\Phi x_0).
\end{equation}
where $\Phi=e^{A}, Q = M_{10}^{-1}$, and
\[
M_{10}=\int_0^1 e^{A(1-\tau)}BB^T e^{A^T(1-\tau)}d\tau
\]
is the controllability Grammian. 
Apparently, it reduces to the standard cost $c(x_0, x_1)=\|x_0-x_1\|^2$ when $A=0, B=1$.

\section{Tracking with optimal mass transport from output measurements}\label{sec:track}
We next extend this connection between optimal transport and optimal control to systems with output measurements.
To this end, assume that the underlying dynamic and output measurements corresponds to the linear system
	\begin{subequations}\label{eq:linearsys}
	\begin{eqnarray}
	\dot x(t)&=&A x(t) +Bu(t)\label{eq:linearsysa}\\
	y(t) &=& C x(t)\label{eq:linearsysb}
	\end{eqnarray}
	\end{subequations}
where $A\in \mR^{n\times n}, B\in \mR^{n\times p}, C\in \mR^{m\times n}$ and $(A,B)$ is a controllable pair.
We seek to track the time varying distribution $\hat \rho_t$, where each particle abides by \eqref{eq:linearsysa}, based on the output distributions $\rho_k=C_\#\hat \rho_k$ at the times $k=0, 1, \ldots, T$. 
Note that only the distribution of the output is available. We don't have access to the information about each particle, namely, the particles are indistinguishable.
For example, with a finite ensemble $(x_1(t), \ldots, x_N(t))$ then $\hat \rho_t=\sum_{i=1}^N \delta(x-x_i(t))\in \cM_+(\mR^n)$  for $t\in [0,T]$ and the outputs are $\rho_k=\sum_{i=1}^N \delta(y-Cx_i(k))\in \cM_+(\mR^m)$ for $k=0,1,\ldots, T$, where $\delta$ is the Dirac delta function and $\cM_+(\mR^n)$ is the set of measures on $\mR^n$.
 For determining identifiability of this problem, see \cite{zeng2016ensemble}. 

We propose to model this tracking problem as the following optimal mass transport problem.
We seek a flow of nonnegative measures $\hat\rho: t\to \cM_+(X)$ that minimize
\begin{subequations}\label{eq:tracking}
\begin{align}
\min_{u,\hat\rho} \quad & \quad\int_{t=0}^T\int_{x\in X} \|u(t,x)\|^2\hat\rho(t,x)dxdt\\
\mbox{subject to } & \quad \frac{\partial\hat\rho}{\partial t}+\nabla \cdot ((Ax + Bu)\hat\rho)=0\\
&\quad \rho_k=C_\#\hat\rho(k,\cdot), \; \mbox{ for } k=0,1,\ldots, T. 
\end{align}
\end{subequations}
Reformulating this using the Kantorovich formulation of the optimal transport problems we arrive at the linear programming problem
\begin{subequations}\label{eq:optimization}
\begin{align}
\min_{\substack{\hat\rho_k \in \mathcal{M}_+(X)\\  \pi_k\in \mathcal{M}_+(X^2)}} \quad & \;\sum_{k=0}^{T-1}\int_{(x_k,x_{k+1})\in X^2}\!\! c(x_k,x_{k+1})d\pi_k(x_k,x_{k+1})\label{eq:optimizationa}\\
\mbox{subject to } \quad & \; \int_{x_{k+1}\in X} d\pi_k(x_k,x_{k+1})=d\hat\rho_k(x_{k})\label{eq:optimizationb} \\
& \; \int_{x_{k}\in X} d\pi_k(x_k,x_{k+1})=d\hat\rho_{k+1}(x_{k+1})\label{eq:optimizationc}\\
&\; \rho_k=C_\#\hat\rho_k, \; \mbox{ for } k=0,1,\ldots, T,\label{eq:optimizationd}
\end{align}
\end{subequations}
where the cost function $c$ is given by \eqref{eq:cost_lineardyn}.
This optimization problem has a dual formulation in terms of continuous test functions on $C(X)$, denoted ${\rm Cont}(C(X))$.

\begin{thm}\label{thm:duality}
The optimization problem \eqref{eq:optimization} is the dual of 
\begin{subequations}\label{eq:dual} 
\begin{align}
\max_{\substack{\phi_k\in {\rm Cont}(C(X))\\k=0,1,\ldots, T}}  \;& \sum_{k=0}^T \int_{C(X)}\phi_{k}(y_k)d\rho_{k}(y_k)\label{eq:duala}\\
\mbox{ subject to}  \;\;  & \sum_{k=0}^T \phi_{k}(Cx_k)\le \sum_{k=1}^T c(x_{k-1},x_k)\label{eq:dualb} \\
 & \mbox{for all } x_k\in \mR^n, \; k =0,\ldots, T,\nonumber
\end{align}
\end{subequations}
and the duality gap is zero.

\end{thm}

\begin{proof}
See the appendix for a proof. 
\end{proof}

The proof is derived by noting that \eqref{eq:optimization} is a multimarginal optimal mass transport problem \cite{pass2011uniqueness} with cost being a function that separates into $T$ terms and each term only depends on two of the dimensions. 

Both the formulations \eqref{eq:optimization} and \eqref{eq:dual}
are linear programming problems that can be discretized and solved using general purpose solvers if the number of states are small. 
The number of variables grows linearly in terms of $T$ and  exponentially in the number of system dimensions, thus they suffer from the curse of dimensionality.
An alternative approach would be testing all associations between measurements and particles, but this would require a combinatorial search which grows exponentially in $T$. Yet another possible approach for addressing this problem is to use heuristics along the lines of K-means clustering, as proposed in \cite{zeng2017sampled}.  

When the number of states is large we might therefore want to restrict the distributions to certain classes. One  such class of particular interest is the Gaussian distributions.  

\section{Gaussian cases}\label{sec:gaussian}	
In this section, we zoom in to the case when all the marginal distributions are Gaussian. We assume that, for all $k=0,1,\ldots, T$, the $k$-th marginal $\rho_k$ of the measurement is a Gaussian distribution with mean $\mu_k$ and covariance $\Sigma_k$. By linearity, the output density tracking problem can be divided into two parts: interpolating the means $\{\mu_k\}$ and interpolating the covariances $\{\Sigma_k\}$. 

Interpolating the means is equivalent to solving \eqref{eq:tracking} for a single particle. It reduces to the optimal control problem
	\begin{subequations}\label{eq:splinegen}
	\begin{align}
	\min_{u} \qquad & \int_{0}^{T} \|u(t)\|^2 dt \\
	\mbox{subject to }\quad  & \dot x(t) = Ax(t)+Bu(t), &&0\le t \le T\\
	 &Cx(k) = \mu_k,  &&k=0,\ldots, T.
	\end{align}
	\end{subequations}
By introducing a Lagrangian multiplier $\lambda(\cdot)$, it is easy to see the optimal control is of the form $u(t) = B^T\lambda(t)$ with $\lambda$ satisfying the dual dynamics
	$
		\dot\lambda(t) = -A^T \lambda(t)
	$
for each interval $t \in (k, k+1)$. For each interval, if we fix $x(k), x(k+1)$, then we can obtain a closed-form expression for the optimal cost, which is
	\[
		(x(k+1)-\Phi x(k))^T Q(x(k+1)-\Phi x(k)).
	\] 
Therefore, a strategy to solve \eqref{eq:splinegen} is first minimizing $u$ over fixed $x(0),\ldots,x(T)$ and then minimizing the result over $x(0),\ldots,x(T)$. 
	

The covariances part is solved using semidefinite programming (SDP). We first minimize the cost with fixed state $x(0), x(1),\ldots,x(T)$ and then minimize the resulting cost over $x(k), k=0,1,\ldots, T$ subject to the constraint that $Cx(k)$ is zero mean Gaussian distribution with covariance $\Sigma_k$. For fixed $x(k), k=0,1,\ldots, T$, the minimum of the cost is given in the quadratic form
	\begin{align*}
	&\sum_{k=0}^{T-1}c(x(k),x(k+1)) \\&=\sum_{k=0}^{T-1}(x(k+1)-\Phi x(k))^T Q(x(k+1)-\Phi x(k)).
	\end{align*}
We then minimize this cost subject to the distribution constraint of the output, which reads as
	\begin{subequations}\label{eq:gaussiancost}
	\begin{eqnarray}
		\min &&\mathbb{E}\{ \sum_{k=0}^{T-1} c(x(k),x(k+1))\}\\
		&&y(k)=Cx(k) \sim \Sigma_k, ~k=0,1,\ldots, T.
	\end{eqnarray}
	\end{subequations}

We notice that in the state space, the problem can be viewed as $T$ separate optimal transport problems. However, these problems are coupled through the constraints on the output. Since the cost function is quadratic, it is not difficult to show that the solution remains Gaussian. Thus, the cost becomes 
	\[
		\sum_{k=0}^{T-1} {\rm Tr}(Q\hat\Sigma_{k+1}+\Phi^TQ\Phi\hat\Sigma_{k}-2Q\Phi S_{k,k+1}),
	\]
where $S_{k,k+1}=\mE\{x(k) x(k+1)^T\}$. 
\begin{thm}
The density tracking Problem \eqref{eq:tracking} for Gaussian marginals with covariances $\{\Sigma_0, \Sigma_1,\ldots, \Sigma_T\}$ has the SDP formulation
	\begin{subequations}\label{eq:SDPa}
	\begin{eqnarray}\label{eq:SDPa1}
		\hspace{-0.2cm}\min_{\hat\Sigma,S}\hspace{-0.5cm} &&\sum_{k=0}^{T-1} {\rm Tr}(Q\hat\Sigma_{k+1}+\Phi^TQ\Phi\hat\Sigma_{k}-2Q\Phi S_{k,k+1})\\\label{eq:SDPa2}
		&& \left[\begin{matrix}
		\hat \Sigma_k & S_{k,k+1} \\ S_{k,k+1}^T & \hat \Sigma_{k+1} 
		\end{matrix}\right] \ge 0,~~k=0,\ldots, T-1\\
		&& C\hat\Sigma_{k}C^T=\Sigma_k,~~k=0,1,\ldots, T.\label{eq:SDPa3}
	\end{eqnarray}
	\end{subequations}
\end{thm}

We remark that it suffice to have constraint \eqref{eq:SDPa2} to guarantee a well-defined covariance matrix for the random vector $(x_0, x_1, \ldots,x_T)^T$. This can be proven constructively using graphical models, see \cite{CheConGeo18}.

Alternatively, to solve \eqref{eq:gaussiancost}, one can choose to minimize the cost over fixed $y(k), k=0,1,\ldots,T$ first and then minimize result over the distribution of $y(k)$. Let $c_y(y(0),y(1),\cdots,y(T))$ be the minimum of the cost for fixed outputs. More precisely,
	\begin{eqnarray*}
	c_y(y(0),\cdots,y(T))&=&\min_x \sum_{k=0}^{T-1} c(x(k),x(k+1))\\
	&&y(k)=Cx(k),~k=0,1,\ldots, T.
	\end{eqnarray*}
Since $c(\cdot,\cdot)$ is quadratic, $c_y$ is a quadratic function of $y=[y(0),y(1),\cdots,y(T)]^T$. Furthermore, it has the form
	\[
		c_y(y) = y^T R y
	\]
for some positive semidefinite matrix $R$. Thus, problem \eqref{eq:gaussiancost} becomes
	\begin{subequations}
	\begin{eqnarray}
		&& \min_y \mE\{\trace (R yy^T)\}
		\\
		&&y(k) ~\sim~ \Sigma_k, ~k=0,1,\ldots, T,
	\end{eqnarray}
	\end{subequations}
or equivalently,
	\begin{subequations}\label{eq:SDPb}
	\begin{eqnarray}
		&& \min_{\Sigma_y\ge 0} \trace (R \Sigma_y)
		\\
		&&\Sigma_y(k,k) = \Sigma_k, ~k=0,1,\ldots, T.
	\end{eqnarray}
	\end{subequations}

\begin{remark}
The two SDP formulations \eqref{eq:SDPa} and \eqref{eq:SDPb} represent two different ways of viewing the density tracking problem. In \eqref{eq:SDPa} we lift the marginals distributions to the state space and solve $T$ separate optimal transport problems. In contrast, \eqref{eq:SDPb} approaches the problem directly by considering the joint distributions of $y(0),y(1),\cdots,y(T)$. In terms of complexity, \eqref{eq:SDPa} grows linearly as the number of time points $T$, while \eqref{eq:SDPb} grows as $T^6$ in the worse case. Therefore, \eqref{eq:SDPa} is better for computational purpose.
\end{remark}

After obtaining the marginal covariances $\{\hat\Sigma_k\}$ for the state variables, we can recover  $\hat\Sigma(t), k\le t\le k+1$ from $\hat\Sigma_k, \hat\Sigma_{k+1}$ in closed-form for each $k=0,\ldots, T-1$, using optimal mass transport theory over linear dynamics \cite{chen2017optimal}.  It follows that the trajectory of the covariances of the output is given by $\Sigma(t)=C\hat\Sigma(t)C^T$. Finally, let $\mu(t)=x(t)$ be the solution to \eqref{eq:splinegen}, we obtain our optimal density flow being a flow of Gaussian distributions with mean $\mu(t)$ and covariance $\Sigma(t)$.


\section{Examples} \label{sec:example}
Two examples are provided to illustrate our framework. In the first example, we explain the tracking of finite number of particles and in the second one, we consider a Gaussian distributions tracking problem.
The implementations are made in Matlab using the package CVX, which is a package for specifying and solving convex programs \cite{grant2008graph, cvx}.

\begin{figure}
\includegraphics[width=0.48\textwidth]{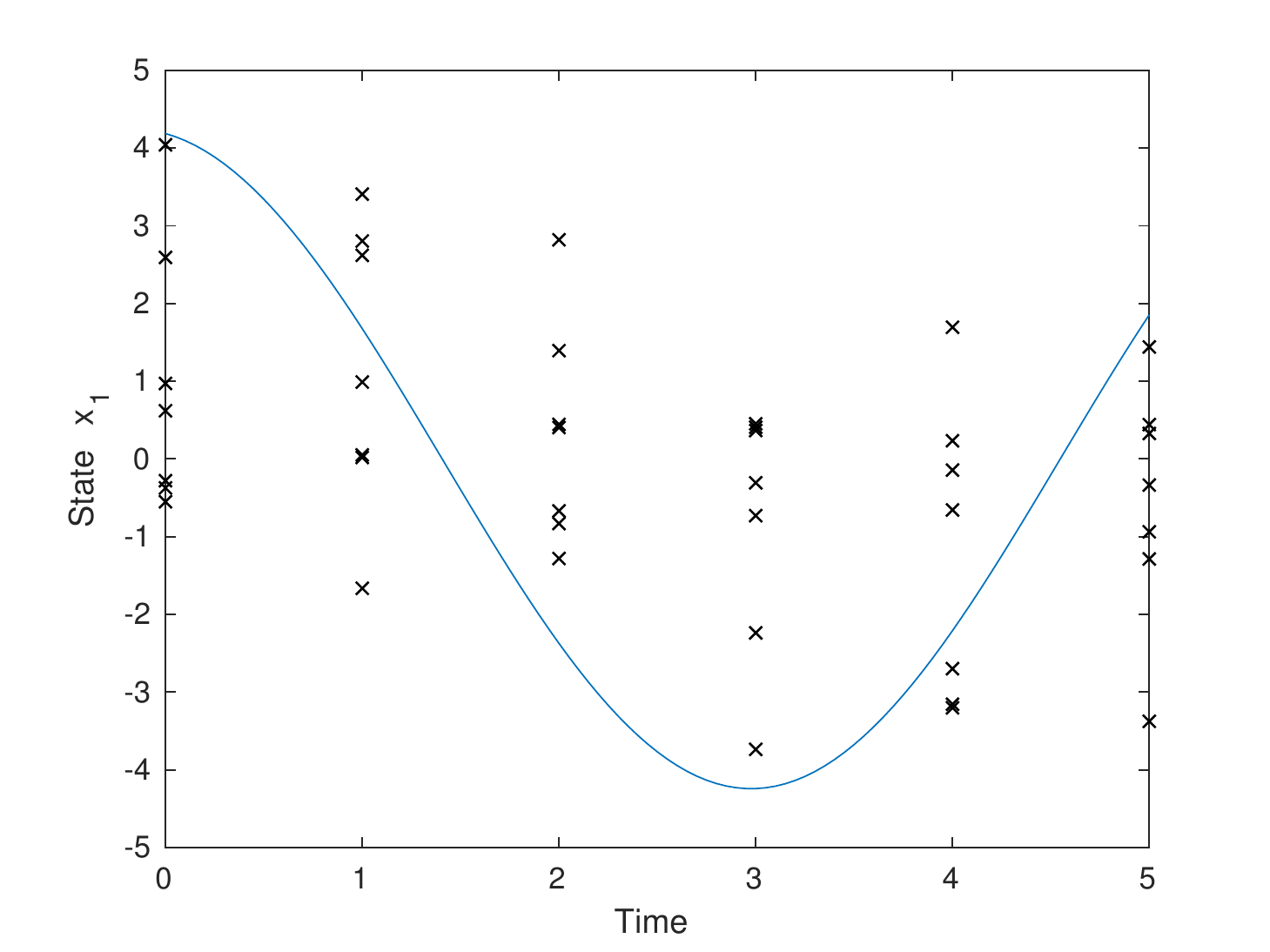}

\caption{Ensemble outputs in example with $N=7$ systems. System dynamic given by \eqref{eq:example_dynamic} and $\sigma=0.1$. Blue trajectory: example of a system trajectory with $\sigma=0$.}
\label{fig:tracking_0}

\end{figure}

\begin{figure}
        \centering
\includegraphics[width=0.48\textwidth]{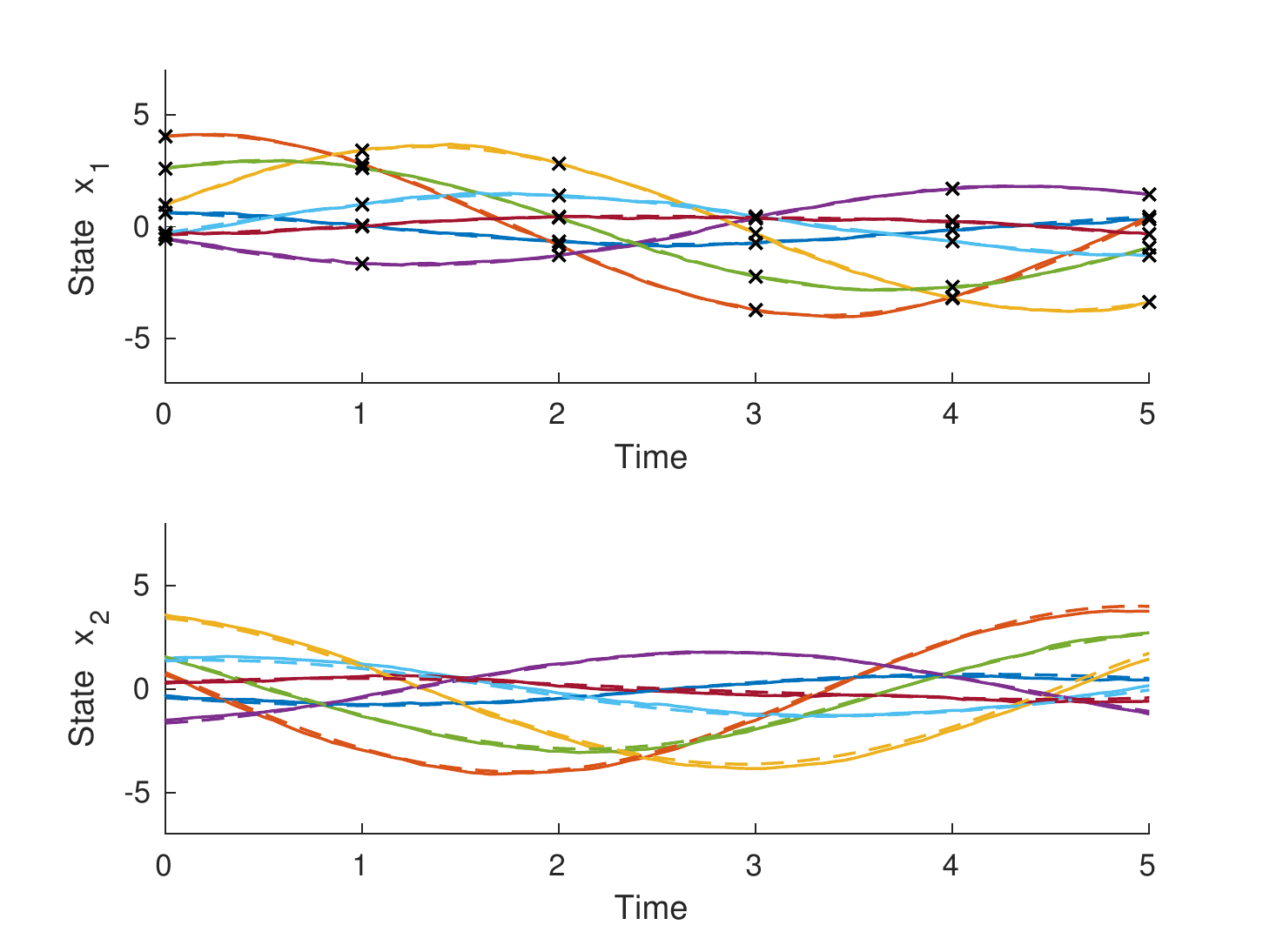}
        \caption{Example with $N=7$ systems to be tracked. System dynamic: \eqref{eq:example_dynamic}. Noise level: $\sigma=0.1$. Available measurement points (x). True system states (solid). Estimated system stares (dashed) Upper figure: State $x_1$.  Lower figure: State $x_2$.}
\label{fig:tracking_1}
\end{figure}

\begin{figure}
    \centering
	\includegraphics[width=0.48\textwidth]{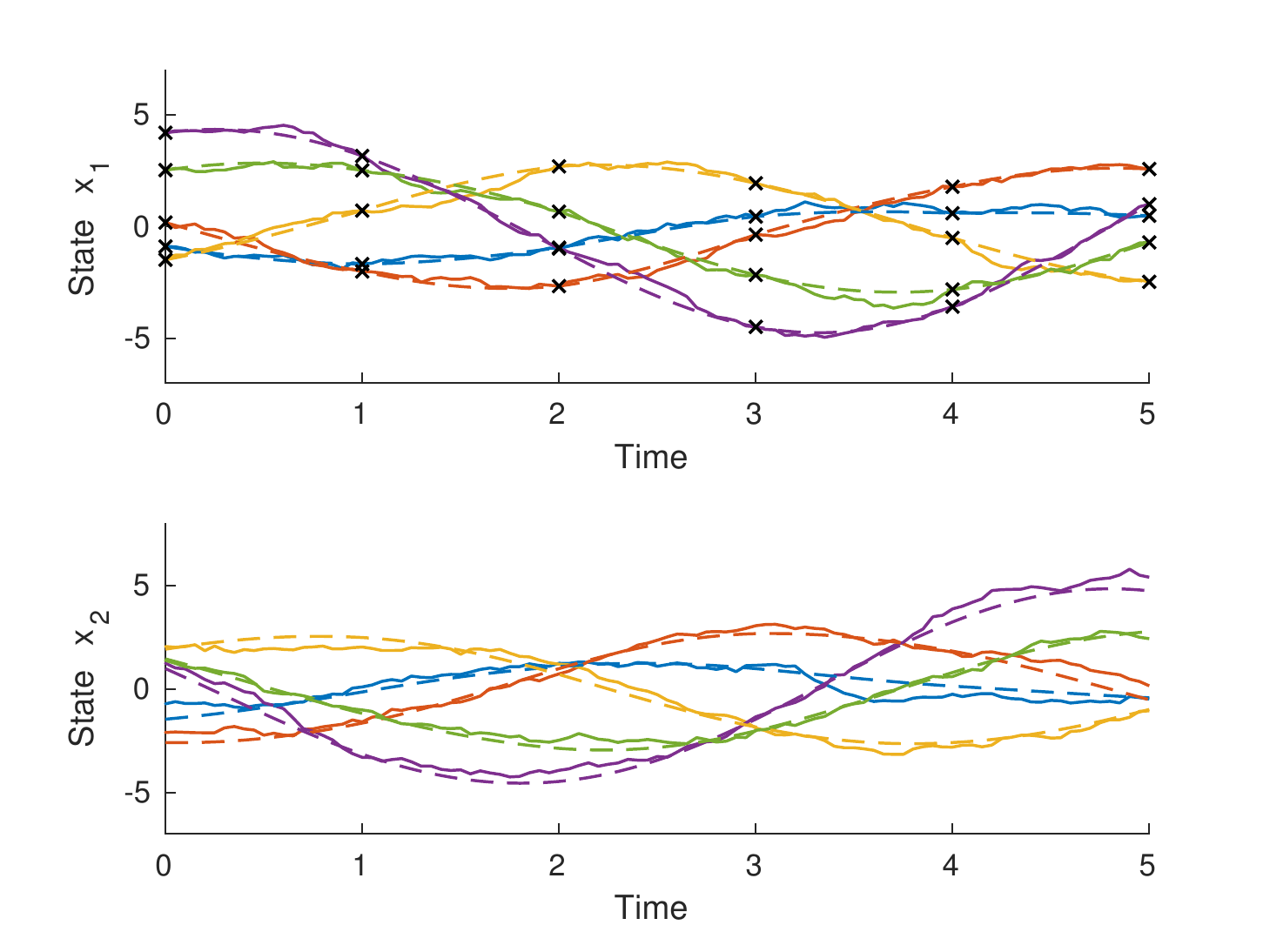}
        \caption{Example with $N=5$ systems to be tracked. Noise level: $\sigma=0.5$. Available measurement points (x). True system states (solid). Estimated system stares (dashed) Upper figure: State $x_1$.  Lower figure: State $x_2$.}
\label{fig:tracking_2} 
\end{figure}

\subsection{Tracking of particles}
We illustrate the tracking of a series of systems with a given system dynamic.
Consider the tracking of $N$ systems with oscillatory dynamics
\begin{align*}
dx(t)&=A x(t) dt + \sigma dw(t)\\
y(t)&=C x(t)
\end{align*}
where the state dynamics is given by
\begin{equation}\label{eq:example_dynamic}
A=\begin{bmatrix}
0&1\\-1&0
\end{bmatrix},
\quad C=\begin{bmatrix}
1&  0
\end{bmatrix}
\end{equation}
and where $dw$ is normalized white Gaussian noise.
We seek to recover the full state information of the systems based on only the unordered outputs, observed at the time point $t=0, 1,\ldots, 5$. 
The problem is solved in the two following examples using the formulation \eqref{eq:optimization}, where the state space is discretized. The discretization of the state space at time $k$ is given by 
\[\{(x_1, x_2)| x_1\in {\rm supp}(\rho_k), x_2\in {\rm linspace} ([-7,7], 150)\},
\]
where ${\rm linspace} ([a,b], N)$ denotes the set of $N$ uniformy spaced points in the interval $[a,b]$.

For the first example we let the number of particles be $N=7$ and the noise level $\sigma=0.1$. Figure \ref{fig:tracking_0} shows the output measurement and an example of a noiseless trajectory corresponding to \eqref{eq:example_dynamic}. We seek to group the measurements corresponding to which system they belong.  
Figure~\ref{fig:tracking_1} shows the reconstruction based on the optimization problem \eqref{eq:optimization}, and both the state estimates correspond well to the true states. Even though there are a few error in the associations, these are between points that are spaced closely together and does not seem to impact the state estimation significantly. 
Next, we consider an example with $N=5$ particles and noise level $\sigma=0.5$. 
Figure~\ref{fig:tracking_2} shows the reconstruction based on the optimization problem \eqref{eq:optimization}. Even though the noise level is fairly large, we are able to achieve good reconstructions of the states.

\subsection{Tracking of Gaussian distributions}
We consider a dynamical system consisting of a simple, possibly high dimensional, first order integrator. The dynamics are governed by
\[
A=\begin{bmatrix}
0&I\\0&0
\end{bmatrix},\; 
B=\begin{bmatrix}
0\\I
\end{bmatrix}, \;
C=\begin{bmatrix}
I&0
\end{bmatrix},
\]
where $I$ is an identity matrix of proper size $n$. When the dimension of the output is $n=1$, we randomly generate several covariances and interpolate them using \eqref{eq:SDPa}. The results are depicted in Figure \ref{fig:oneD} for $T=5,10$. It can be observed that the interpolated curves are smooth. Similar results hold for high dimensional setting. Figure \ref{fig:marginals} depicts several Gaussian distributions with different means and covariances we want to track. The tracking result is shown in Figure \ref{fig:twoN4}, which is a natural and smooth interpolation of the observations.

\begin{figure}[t!]
    \begin{subfigure}[t]{0.26\textwidth}
\centering
\includegraphics[width=\textwidth]{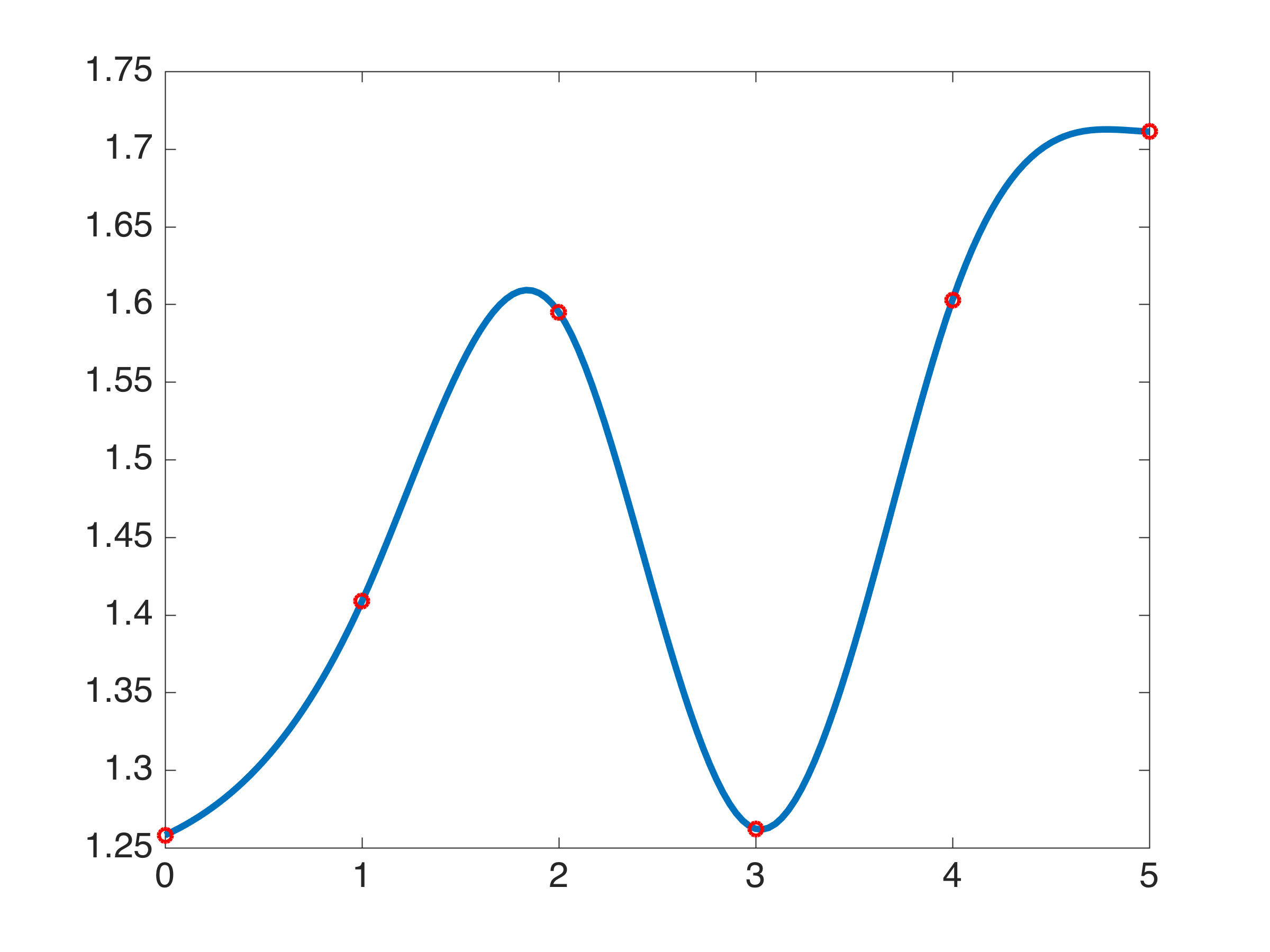}
\subcaption{$T=5$}
\label{fig:N5}
    \end{subfigure}%
    ~ 
    \begin{subfigure}[t]{0.26\textwidth}
\centering
\includegraphics[width=\textwidth]{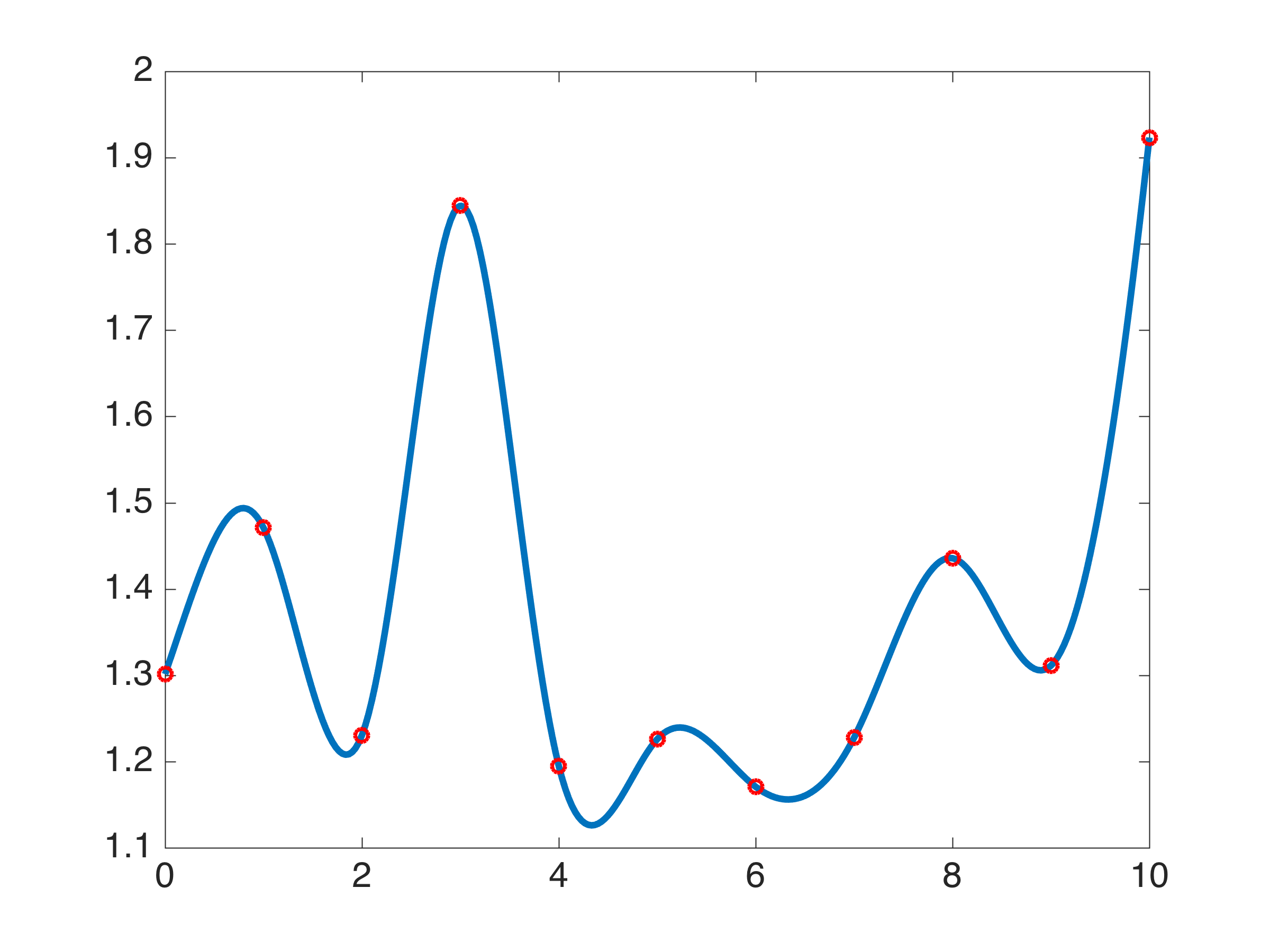}
\subcaption{$T=10$}
\label{fig:N10}
    \end{subfigure}
    \caption{Interpolation of covariances}
\label{fig:oneD}
\end{figure}


\section{Conclusions}\label{sec:conclusion}
A framework of tracking the states of indistinguishable particles with linear dynamics using output measures is presented. The measurements are the distributions of the output at several time points. Our framework relies on a recent development of optimal mass transport theory with prior dynamics \cite{chen2017optimal}. In the special case with Gaussian marginals, our problem has a SDP formulation and can be solved efficiently. Developing fast algorithms for the general cases will be a future research topic.

\begin{figure}
\centering
\includegraphics[width=0.5\textwidth]{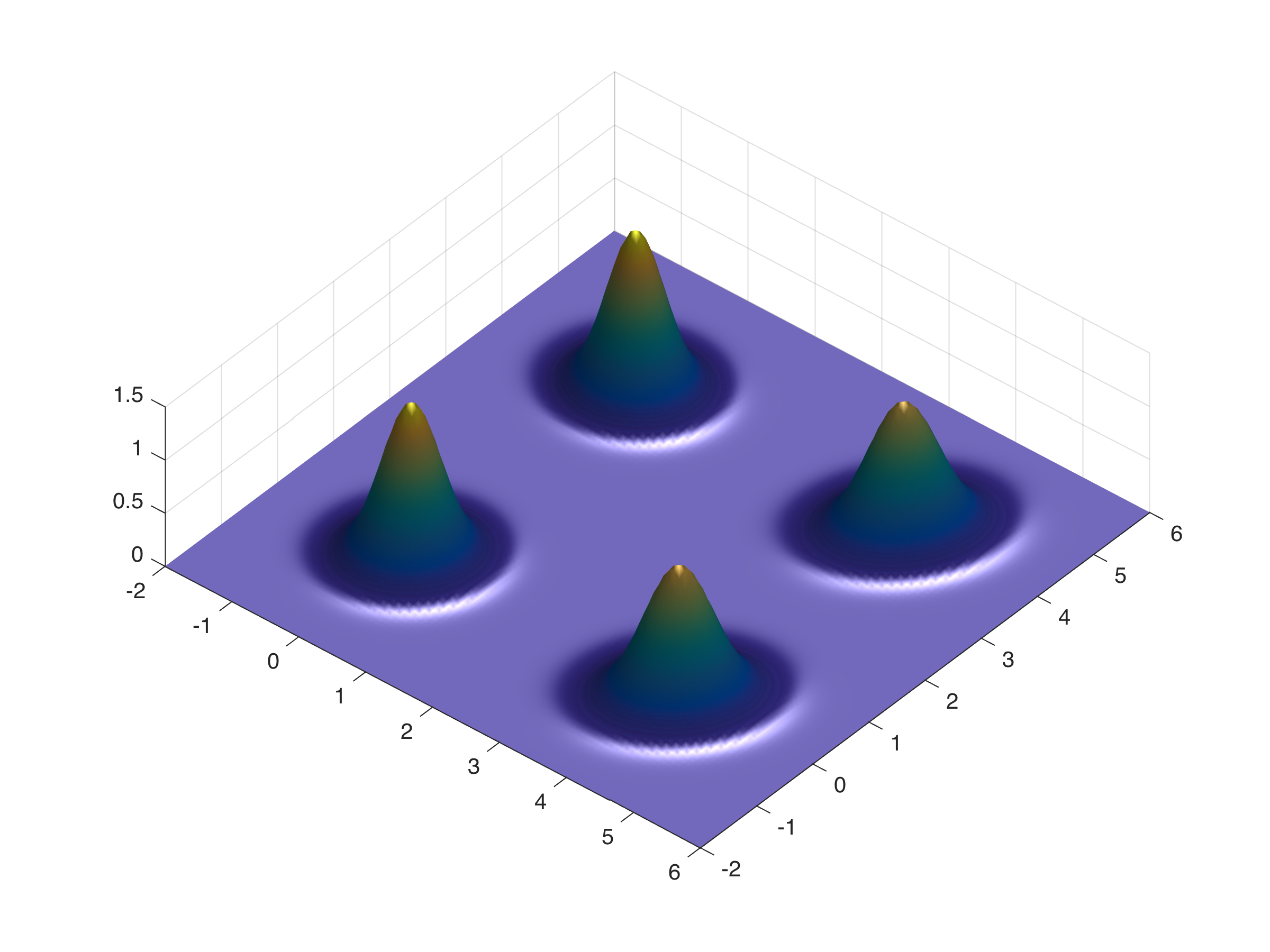}
\caption{Marginal distributions}
\label{fig:marginals}
\end{figure}
\begin{figure}
\centering
\includegraphics[width=0.5\textwidth]{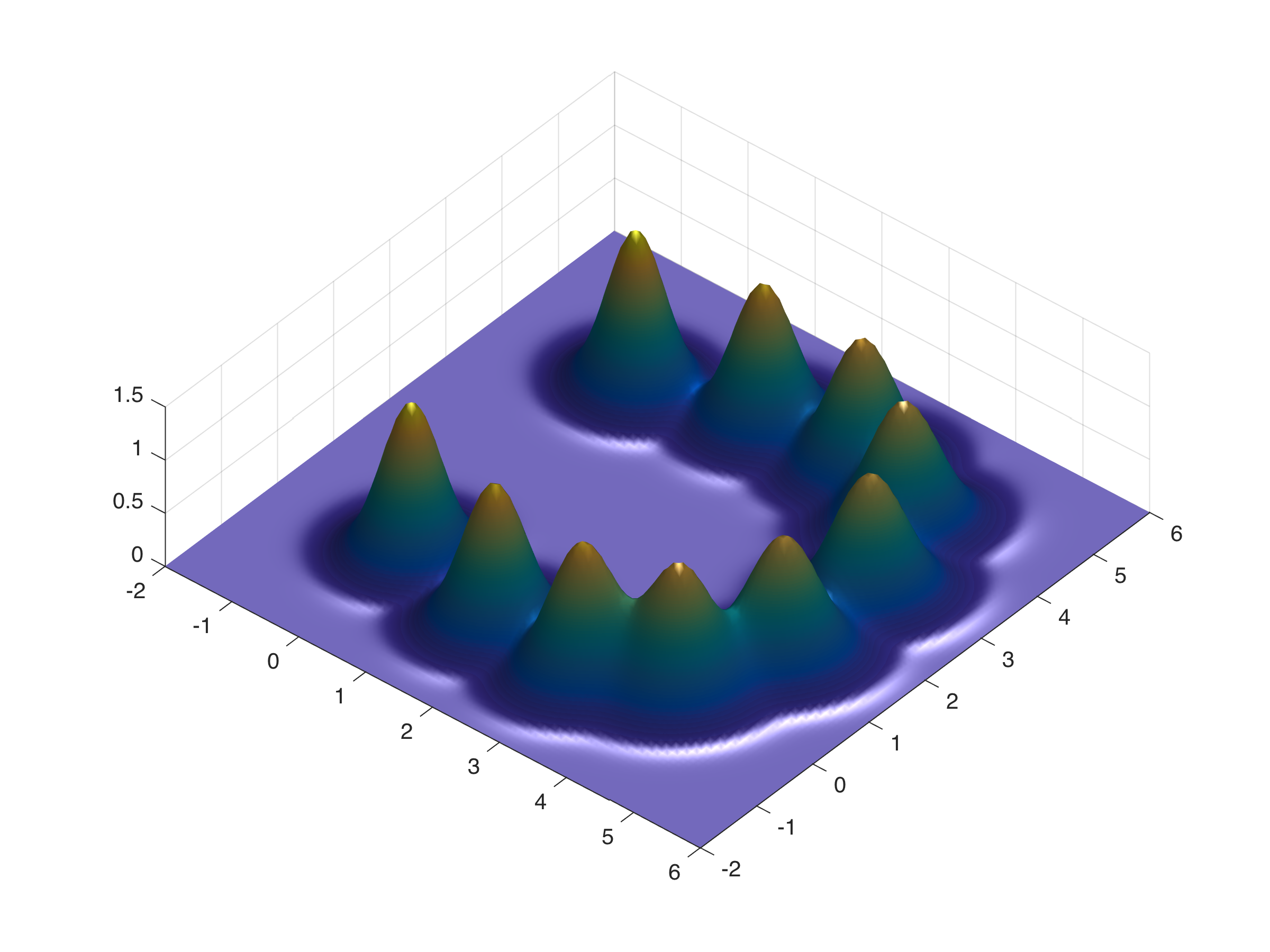}
\caption{Interpolation of covariances: $T=3$}
\label{fig:twoN4}
\end{figure}

\appendix

\subsection{Proof of Theorem~\ref{thm:duality}}

The Lagrangian of the problem \eqref{eq:dual}, where the constraint \eqref{eq:dualb} is relaxed, is given by
\begin{align*}
&L(\{\phi_k\}_{k=0}^T, M)=\sum_{k=0}^T \int_{C(X)}\phi_{k}(y_k)d\rho_{k}(y_k)\\&\quad +\int_{X^{T+1}}\left(\sum_{k=1}^T c(x_{k-1},x_k)-\sum_{k=0}^T \phi_{k}(Cx_k))\right)dM(\bx)
\end{align*}
where $\bx=\begin{pmatrix}x_0,\ldots, x_T\end{pmatrix}\in X^{T+1}$ \cite[Chapter~8]{luenberger1997optimization}. Note that since the constraint \eqref{eq:dualb} is defined by continuous functions, then the dual variable is a nonnegative measure $M\in \cM_+(X^{T+1})$.
Denote the one dimensional and two-dimensional marginals of $M$ by
\begin{subequations}\label{eq:marginals}
\begin{align}
d\hat \rho_k(x_k)&=\int_{x_0,\ldots, x_{k-1},x_{k+1}, \ldots, x_{T}\in X}dM(\bx),\label{eq:marginalsa}\\
d\pi_k(x_{k}, x_{k+1})&=\int_{x_0,\ldots, x_{k-1},x_{k+2}, \ldots, x_{T}\in X}dM(\bx).\label{eq:marginalsb}
\end{align}
\end{subequations}
For a fixed $M$, the maximum of $L(\{\phi_k\}_{k=0}^T, M)$ with respect to $\phi_k$ is finite only if 
\[
\int_{C(X)}\phi_{k}(y_k)d\rho_{k}(y_k)=\int_{X} \phi_{k}(Cx_k)) \hat \rho_k(x_k)
\] 
for any continuous functional $\phi_k: C(X)\mapsto \mR$. From the Rietz representation theorem (see e.g. \cite{luenberger1997optimization}) and the definition \eqref{eq:masspreserving} of mass preserving maps this is equivalent with $C_\# \rho_k=\hat \rho_k$.
Also, using \eqref{eq:marginalsb} the remaining second term of $L(\{\phi_k\}_{k=0}^T, M)$ can be written as 
\begin{equation}\label{eq:objective_function}
\sum_{k=0}^{T-1}\int_{(x_{k},x_{k+1})\in X\times X} c(x_{k},x_{k+1})d\pi_k(x_{k},x_{k+1}).
\end{equation}

Next, note that for any $M\in \cM_+(X^{T+1})$ the marginals \eqref{eq:marginals} satisfy the constraints \eqref{eq:optimizationb} and \eqref{eq:optimizationc} in the optimization problem \eqref{eq:optimization}. Conversely we would like to show that for any set of marginals satisfying the constraints \eqref{eq:optimizationb}-\eqref{eq:optimizationc} for $k=0,1,\ldots, T-1$ there is distribution $M\in\cM_+(X^{T+1})$ that satisfies \eqref{eq:marginals}.
Thus, let the one-dimensional marginals $\hat \rho(x_k)$ and two-dimensional marginals $\pi_k(x_{k}, x_{k+1})$ satisfy \eqref{eq:optimizationb} and \eqref{eq:optimizationc} for $k=0,1,\ldots, T-1$. 
From equations \eqref{eq:optimizationb}-\eqref{eq:optimizationc} it follows that the total mass $\hat \rho_k(X)=:\kappa$ is the same for all $k=0,1,\ldots, T$. Normalizing the total mass to $1$, we may consider the Markov chain with transition probabilities from time $k$ to $k+1$ specified by the joint distribution $\pi(x_{k}, x_{k+1})/\kappa$. Consequently if we let $M/\kappa$ be the measure corresponding to the joint probabilities of the Markov chain, then $M$ satisfies \eqref{eq:marginals} (see e.g., \cite{sharpe1988general}). 

To summarize, we have shown that the dual functional of \eqref{eq:dual} is given by
\[
\max_{\{\phi_k\}_{k=0}^T} L(\{\phi_k\}_{k=0}^T, M)=
\begin{cases}
\eqref{eq:objective_function} & \mbox{if } d\rho_k \mbox{ satisfies } \eqref{eq:optimizationd}\\
\infty & \mbox{otherwise,}
\end{cases}
\]
where $d\pi_k$ are the marginals given by \eqref{eq:marginalsb}. The functions $d\rho_k$ and $d\pi_k$ are marginals \eqref{eq:marginals} for some $M\in\cM_+(X^{T+1})$ if and only if 
the marginals satisfies the constraints \eqref{eq:optimizationb}-\eqref{eq:optimizationc}, hence the dual of \eqref{eq:dual} is given by \eqref{eq:optimization}. Furthermore, since \eqref{eq:dual} has a feasible interior point, the duality gap is zero \cite[p. 217]{luenberger1997optimization}. \hfill $\blacksquare$

\balance

\bibliographystyle{plain}

\bibliography{refs,bib_johan}

\end{document}